\documentclass{article}
\usepackage{amsmath,amsthm,amssymb,mathptmx,graphicx,wrapfig}
\DeclareMathOperator{\KA}{\mathit{KM}}
\DeclareMathOperator{\KP}{\mathit{K}}
\DeclareMathOperator{\m}{\mathbf{m}}
\DeclareMathOperator{\A}{\mathbf{a}}
\DeclareMathOperator{\T}{\mathbf{t}}
\newtheorem{theorem}{Theorem}
\newtheorem{proposition}[theorem]{Proposition}
\newtheorem{lemma}[theorem]{Lemma}

\title{The sum $2^{\KA(x)-\KP(x)}$ over all prefixes $x$ of some binary sequence can be infinite}
\author{Mikhail Andreev, Akim Kumok}

\begin{document}
\maketitle

\begin{abstract}
We consider two quantities that measure complexity of binary strings: $\KA(x)$ is defined as the minus logarithm of continuous a priori probability on the binary tree, and $\KP(x)$ denotes prefix complexity of a binary string $x$. In this paper we answer a question posed by Joseph Miller and prove that there exists an infinite binary sequence $\omega$ such that the sum of $2^{\KA(x)-\KP(x)}$ over all prefixes $x$ of $\omega$ is infinite. Such a sequence can be chosen among characteristic sequences of computably enumerable sets.
\end{abstract}

\section{Introduction}

Algorithmic information theory tries to define the notion of \emph{complexity} of a finite object and the related notion of its \emph{a priori probability}. Both notions have different versions, and many of these versions can be used to define algorithmic randomness. To explain why the result of this paper could be interesting, let us start with a short survey of these notions and related results; for the detailed exposition of the related definitions and results see, e.g.,~\cite{kolmbook,ShenNotes}.

A notion of prefix complexity was introduced by Levin (see \cite{LevinThesis,LevinPpi,GacsSymmetry} and later by Chaitin~\cite{ChaitinPrefix} (in different forms). Let $D$ be a computable function whose arguments and values are binary strings. This function is called \emph{prefix-free} if its domain is prefix-free, i.e., does not contain both a string and its non-trivial prefix. Define $\KP_D(x)$ the minimal length of $p$ such that $D(p)=x$. Among all functions $\KP_D$ for all computable prefix-free $D$ there exists a mininal one (up to $O(1)$ additive term); one of them is fixed and called $\KP(x)$, the prefix-free complexity of $x$. (Another version, which gives the same function $\KP$ with $O(1)$-precision, uses \emph{prefix-stable} functions $D$: this means that if $D(x)$ is defined, then $D(xz)$ is defined and equals $D(x)$ for all $z$).

The prefix complexity is closely related with the \emph{discrete a priori probability}~\cite{LevinThesis,LevinPpi,ChaitinPrefix}. Consider a non-negative total real function $m$ defined on binary strings. We call $m$ a \emph{discrete semimeasure} if $\sum_x m(x)\le 1$. We say also that $m$ is \emph{lower semicomputable} if $m(x)$ can be represented as a limit of a non-decreasing sequence $M(x,0),M(x,1),\ldots$ where $M$ is a non-negative total function of two arguments with rational values. Levin  introduced this notion and showed that there exist a maximal (up to $O(1)$-factor) lower semicomputable semimeasure, and this semimeasure is equal to $2^{-\KP(x)+O(1)}$. We fix some maximal lower semicomputable semimeasure, call it the \emph{discrete a priori probability} (see below about the continuous a priori probability), and denote it in the sequel by~$\mathbf{m}(x)$.

Discrete lower semicomputable semimeasures are exactly the output distributions of probabilistic machines without input that produce their output at once (say, write a binary string and then terminate). We can also consider probabilistic machines without input that produce their output bit by bit (and never terminate explicitly, though it may happen that they produce only finitely many output bits). The output distributions of such machines are described by lower semicomputable \emph{continuous semimeasures} (=semimeasures on a binary tree), introduced in~\cite{ZvonkinLevin}.  By a continuous semimeasure we mean a non-negative total function $a$ that is defined on binary strings and has the following two properties:
\begin{itemize}
\item $a(\Lambda)=1$, where $\Lambda$ is an empty string;
\item $a(x)\ge a(x0)+a(x1)$ for every string $x$.
\end{itemize}
There exists a maximal (up to $O(1)$-factor) lower semicomputable continuous semimeasure; it is called the \emph{continuous a priori probability} and is denoted by $\A(x)$ in the sequel. The quantity $-\log_2 \A(x)$ is ofter called \emph{a priori complexity} and sometimes denoted $\KA(x)$.

Now we have defined all the quantities involved in our main result, but to explain its informal meaning we should say more about algorithmic randomness. (These explanations are not needed to understand the statement and the proof of the main result, so the reader may jump to the next section.)

The notion of a random sequence was introduced by Martin-L\"of in 1966 (see~\cite{MartinLofDefinition}). Let $P$ be a computable measure on the Cantor space $\Omega=\{0,1\}^\infty$ of infinite binary sequences; this means that the values $P(x\Omega)$ of the cylinders (here $x\Omega$ is the set of all infinite extensions of a binary string $x$) can be effectively computed with arbitrary precision. An \emph{effectively open} subset of $\Omega$ is a union of a (computably) enumerable set of cylinders. A \emph{Martin-L\"of test} (with respect to $P$) is an uniformly enumerable decreasing sequence of effectively open sets 
    $$ U_1\supset U_2 \supset U_3\supset\ldots $$
such that $P(U_i)\le 2^{-i}$. A sequence $\omega\in\Omega$ \emph{passes} this test if it does not belong to $\bigcap_i U_i$. \emph{Martin-L\"of random sequences} are sequences that pass all tests (with respect to $P$).

In 1970s Levin and Gacs found an useful reformulation of this definition in terms of \emph{randomness deficiency} function. Consider a lower semicomputable function $t$ on the Cantor space with non-negative real values (possible infinite). Lower semicomputability means that this the set $\{\omega|t(\omega)> r\}$ is effectively open for all positive rational $r$ uniformly in $r$. A \emph{Levin--Gacs test} with respect to $P$ is such a function with finite integral $\int t(\omega)\,dP(\omega)$. For a given $P$ there is a maximal (up to $O(1)$-factor) Levin--Gacs test; Martin-L\"of random sequnces are exactly the sequences for which this test is finite.

There is a formula\footnote{It goes back to P.~G\'acs paper~\cite{GacsExactExpression}, but G\'acs used a different and rather cumbersome notation there. See~\cite{longpaper} for the detailed exposition.} that expresses a maximal test in terms of a priori probability: 
  $$\T(\omega)=\sum_{x\sqsubset \omega}\frac{\m(x)}{P(x)};$$
here $x\sqsubset\omega$ means that binary string $x$ is a prefix of an infinite binary sequence $\omega$; note that $\T$ in the left-hand side depends on $P$ (though this is not reflected in the notation). Moreover, the sum in this formula can be replaced by the supremum.

For the uniform Lebesgue measure on the Cantor space this result can be rewritten as follows:
  $$\T(\omega)=\sum_n 2^{n-\KP(\omega_1\ldots\omega_n)}=2^{\sup_n( n-\KP(\omega_1\ldots\omega_n))}.$$
This equation implies both Schnorr--Levin criterion of randomness (see~\cite{SchnorrCrit,LevinCrit}; its version with prefix complexity saying $\omega$ is Martin-L\"of random with respect to the uniform measure iff $n-\KP(\omega_1\ldots\omega_n)$ is bounded, is mentioned in~\cite{ChaitinPrefix}) and the Miller--Yu \emph{ample excess} lemma (\cite{MillerYu}, section~2) saying that the sum in the right hand side is finite for random $\omega$. 

There were many attempts to generalize a notion of randomness to a broader class of distributions, not only computable measures. The notion of uniform test (a function of two arguments: a sequence and a measure) was introduced by Levin (see~\cite{LevinCrit,LevinUniform}); it was used to define \emph{uniform randomness} with respect to arbitrary (not necessarily computable) measure $P$. Levin proved that there exists a \emph{neutral measure} $N$ such that every sequence is uniformly random with respect to $N$ (and even has uniform randomness deficiency at most $1$), see~\cite{longpaper} for the exposition of these results.

One could also try to extend the definition to continuous lower semicomputable \emph{semi}measures (a broader class than computable  measures where $a(x)=a(x0)+a(x1)$). Such a semimeasure is an output distribution of a probabilistic machine and one may ask which sequences are ``plausible outcomes'' for such a machine. In this case there is no universally accepted definition; one of the desirable properties of such a definition is that every sequence should be random with respect to continuous a priori probability $\A(\cdot)$ (that corresponds to a probabilistic machine for which we do not have any a priori information).

One of the possibilities would be to use Gacs' formula as a definition and say that a sequence $\omega$ is random with respect to a continuous semimeasure $A$ if the sum $\sum_{x\sqsubset\omega} \m(x)/A(x)$ is finite, or if the supremum $\sup_{x\sqsubset\omega} \m(x)/A(x)$ is finite.  If $A$ is the continuous a priori probability, the supremum is always finite (and uniformly bounded: it is easy to see that $\m(x)/\A(x)\le O(1)$ for all $x$). Moreover, in 2010 Lempp, Miller, Ng and Turetsky (unpublished; we thank J.~Miller who kindly provided a copy of this note) have shown that for every $\omega$ the ratio $\m(x)/\A(x)$ tends to zero for prefixes $x\sqsubset\omega$ (though it is $\Theta(1)$, say, for strings of the form $0^n1$).

In this paper we show (Theorem~\ref{th:main} in Section~\ref{sec:main}) that this result cannot be strengthened to show that the sum of $\m(x)/\A(x)$ along every sequence is bounded. So the first of the suggested definitions of randomness with respect to semimeasure (with the sum instead of supremum) differs from the second one: not all sequences are random with respect to $\A$, according to this definition.

It would be interesting to understand better for which sequences the sum $$\sum_{x\sqsubset\omega}\m(x)/\A(x)$$ is finite. Are they related somehow to K-trivial sequences (where $\m(x)$ is equal to $\m(|x|)$ up to $O(1)$-factor)? We do not know the answer; we can show only (see Section~\ref{sec:enumerable}) that one can find a computably enumerable set whose characteristic sequence has this property.

Our result about the sum of $\m(x)/\A(x)$ is of computational nature: if we allow more computational power for $\A(x)$, the sum becomes finite, as the following simple proposition shows.

\begin{proposition}\label{prop:offline}
Let $\A'=\A^{\mathbf{0}'}$ be the relativized continuous a priori probability using $\mathbf{0}'$ as an oracle. Then the sum
$\sum_{x\sqsubset\omega}{\frac{\m(x)}{\A'(x)}}$ is bounded for all $\omega$ by a constant (not depending on $\omega$).
\end{proposition}

\begin{proof}
It is enough to construct a $\mathbf{0}'$-computable \emph{measure} $a'$ such that $\sum_{x\sqsubset\omega}{\frac{\m(x)}{a'(x)}}\le 1$ for all $\omega$. (Then we can note that $\A'(x)$ is an upper bound for $a'$.) One can describe such a measure explicitly. Let us add all the a priori probabilities of all strings $x$ that start with $0$ and with $1$:
$$
     M_0=\sum_u \m(0u); \quad M_1=\sum_u \m(1u).
$$
(Note that $M_0+M_1+\m(\Lambda)\le 1$, where $\Lambda$ denotes the empty string, the root of the tree.) Now let us split $1$ into $a'(0)+a'(1)$ in the same proportion, i.e., let
$$
a'(0)=\frac{M_0}{M_0+M_1},\quad a'(1)=\frac{M_1}{M_0+M_1}.
$$
Then we continue in the same way, splitting $a'(0)$ into $a'(00)$ and $a'(01)$ in the proportion $M_{00}:M_{01}$, and so on. Here $M_z$, defined for every string $z$, denotes the sum $\sum_u \m(zu)$.

The numbers $M_z$ are lower semicomputable, so they are $\mathbf{0}'$-computable (and positive), and the measure $a'$ is well defined and $\mathbf{0}'$-computable. It remains to check that it is large enough, so the sum in question is bounded by $1$.

It is enough to prove this bound for finite sums (when only vertices below some level $N$ are considered), so we can argue by induction and assume that the similar statement is true for the left and right subtrees of the root,  with appropriate scaling. \footnote{The summation is stopped at the same level $N$, so the tree height is less by $1$ and we can apply the induction assumption. The base of induction is trivial: in the root the ratio $m/a$ is at most $1$ for evident reasons.} The sum of $m(x)$ in the left subtree is bounded by (actually, is equal to) $M_0$, instead of $1$ in the entire tree; the sum in the right subtree in bounded by $M_1$. On the other hand, the values of $a'$ at the roots of these trees, i.e., $a'(0)$ and $a'(1)$, are also smaller. So the induction assumption says that for each path in the left subtree the sum of $\m(x)/a'(x)$ is bounded by $M_0/a'(0)$, and for each path in the right subtree the sum is bounded by $M_1/a'(1)$. Therefore, it remains to show that
  $$
\frac{M_0}{a'(0)}+\m(\Lambda) \le 1, \qquad 
\frac{M_1}{a'(1)}+\m(\Lambda) \le 1. 
  $$

Recall that we defined $a'(0)$ and $a'(1)$ in such a way that they are proportional to $M_0$ and $M_1$ respectively, and the sum $a'(0)+a'(1)=1$. So the both fractions in the last formula are equal to $M_0+M_1$, and it remains to note that $M_0+M_1+\m(\Lambda)$ is the sum of $\m(x)$ over all strings $x$ and is bounded by $1$.

\end{proof}

\textbf{Remark}. Laurent Bienvenu noted that this (simple) computation can be replaced by references to some known facts and techniques. Namely, we know that there exists a neutral measure $N$ such that every binary sequence $\omega$ has uniform deficiency at most $1$ with respect to $N$. This deficiency can be rewritten as $\sum_{x\sqsubset\omega} \frac{\m(x|N)}{N(x)}$ (see~\cite{longpaper} for details). Using low-basis argument, we can choose a $\mathbf{0}'$-computable neutral measure $N$; then  $\A'$ is greater that this $N$. And (in any case) $\m(x|N)$ is greater than $\m(x)$, so we get a desired result.

\section{Main result and the proof sketch: the game argument}
\label{sec:main}

\begin{theorem}\label{th:main}
There exists an infinite binary sequence $\omega$ such that 
   $$
\sum_{x\sqsubset\omega} \frac{\m(x)}{\A(x)}  =\infty.
   $$
\end{theorem}

This is the main result of the paper. The proof uses (now quite standard) game technique. In this section we describe some infinite game and show how the main result follows from the existence of a computable winning strategy for one of the players (called Mathematician, or \textbf{M}) in this game. Then, in Section~\ref{sec:finite-use} we reduce this game to a finite game (more precisely, to a class of finite games), and show that if all these games uniformly have a computable winning strategy for \textbf{M}, then the infinite game has a computable winning strategy. Finally, in Section~\ref{sec:finite-win} we construct (inductively) winning strategies for finite games. (This will be the most technical part of the proof: we even need to compute some integral!)

Let us describe an infinite game with full information between two players, the Mathematician (\textbf{M}) and the Adversary (\textbf{A}). This game is played on an infinite binary tree. 

Mathematician assigns some non-negative rational weights to the tree vertices (=binary strings). Initially all the weights are zeros; at each move \textbf{M} can increase finitely many weights but cannot decrease any of them. The total weight used by \textbf{M} (the sum of her weights) should never exceed $1$. (We may assume that \textbf{M} loses the game immediately if her weights become too big.)  The current \textbf{M}'s weight of some vertex $x$ will be denoted by $m(x)$, so the requirement says that $\sum_x m(x)\le 1$ at any moment of the game (otherwise \textbf{M} loses immediately).

Adversary also assigns increasing non-negative rational weights to the tree vertices. Initially all they are zeros, except for the root weight which is $1$. But the condition is different: for every vertex $x$ the inequality $a(x0)+a(x1)\le a(x)$ should be true.  Informally, one can interpret $a(x)$ as a (pre)flow that comes to vertex $x$. The flow $1$ arrives to the root. From the root some parts $a(0)$ and $a(1)$ are shipped to the left and right sons of the root (while the remaining part $1-a(0)-a(1)$ is reserved for future use. At the next level, e.g., in the vertex $0$, the incoming flow $a(0)$ is split into $a(00)$, $a(01)$ and the (non-negative) reserve $a(0)-a(00)-a(01)$, and so on. As the time goes, the incoming flow (from the father) increases, and it can be used to increase the outgoing flow (to the sons) or kept as a reserve. Again, if \textbf{A} violates the restriction (the inequality $a(x0)+a(x1)\le a(x)$), she loses immediately.

One may assume that the players alternate, though it is not really important: the outcome of the (infinite) game is determined by the limit situation, and postponing some move never hurts (and even can simplify the player's task, since more information about the opponent's moves is then available). We say that \textbf{M} wins if there exist a branch in the tree, an infinite binary sequence $\omega$, such that 
    $$
\sum_{x\sqsubset \omega} \frac{m(x)}{a(x)}=\infty,
    $$
where $m(x)$ and $a(x)$ are limit values of the \textbf{M}'s and \textbf{A}'s weights respectively. One should agree also what happens if some values are zeros. It is not really important since each of the players can easily make her weights positive. However, it is convenient to assume that $m/0=\infty$ for $m\ne 0$ and $0/0=0$.

Now the game is fully defined. Since all the moves are finite objects, one can speak about computable strategies. The following lemma is the main step in the proof of Theorem~\ref{th:main}.

\begin{lemma}\label{lem:infinite}
\textbf{M} has a computable winning strategy in this game.
\end{lemma}

The proof of this lemma will be given in the next two sections. In the rest of this section we explain how the statement of the lemma implies Theorem~\ref{th:main}. This is a standard argument useg in all the game proofs. Consider an ignorant Adversary who does not even look on our (Mathematician's) moves, and just enumerates from below (lower semicomputes) the values of the continuous a priori probability $\A(x)$. (They are lower semicomputable; some additional care is needed to ensure that $a(x)\ge a(x0)+a(x1)$ is true not only for the limit values, but for approximations at every step, but this is done in a standard way, we can increase $a(x)$ going from the leaves to the root.) 

The actions of \textbf{A} are computable. Let \textbf{M} uses her computable winning strategy against such an adversary. Then \textbf{M}'s behavior is computable, too. So the limit values of $m(x)$ form a lower semicomputable function, and the winning condition guarantees that $\sum_{s\sqsubset \omega} m(x)/\A(x)$ is infinite for some sequence $\omega$. It remains to note that the discrete a priori probability $\m(x)$ is an upper bound (up to $O(1)$-factor) for every lower semicomputable function $m(x)$.

\section{Finite games are enough}\label{sec:finite-use}

To construct the winning strategy for \textbf{M} in the infinite game described in the previous section, we combine winning strategies for finite games of similar nature. A finite game is determined by two parameters $N$ and $k$; the value of $N$ is the height of the finite full binary tree on which the game is played, and $k$ is the value of the sum that \textbf{M} should achieve to win the game. Here $N$ is a positive integer, and $k\ge 1$ is a rational number.

Initially all vertices (=all strings of length at most $N$) have zero $a$- and $m$-weights, except for the root that has unit $a$-weight: $a(\Lambda)=1$. The players alternate; at every move each player may increase her weights (rational numbers), but both players should obey the restrictions: the sum of $m$-weights should not exceed $1$; for every $x$ that is not a leaf the inequality $a(x)\ge a(x0)+a(x1)$ should be true; the value of $a(\Lambda)$ remains equal to $1$.  The position of a game is \emph{winning} for \textbf{M} if there exists a leaf $w$ such that the sum $\sum_{x\sqsubset w} m(x)/a(x)$ is at least $k$. Otherwise the position is winning for \textbf{A}. Each player, making a move, should create a winning position (for her), otherwise she loses the game. (She may also lose the game by violating the restrictions for her moves.)

\begin{lemma}\label{lem:finite}
For every positive rational $k$ there exists some $N$ and a winning strategy for \textbf{M} that guarantees that \textbf{M} wins after a bounded number of steps. (The bound  depends on $k$, but not on \textbf{A}'s moves.) The value of $N$ and the strategy are computable given~$k$.
\end{lemma}

\begin{wrapfigure}{r}{0pt}
   \includegraphics{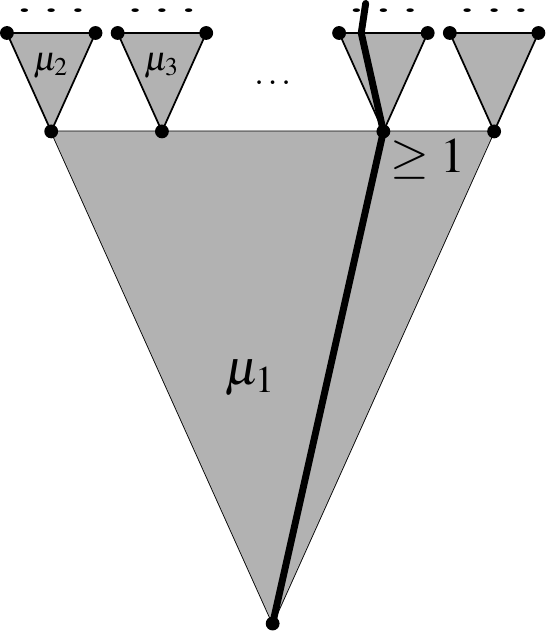}
\caption{Finite subtrees chosen inside an infinite binary tree. On the subtrees \textbf{M} applies a winning strategy for a finite game using quotas $\mu_1,\mu_2,\ldots$, and achieves sum $1$ in every subtree. }\label{subtrees}
\end{wrapfigure}

The proof of this lemma will be given in the next section. In the rest of this section we show how we can use winning strategies for finite games to win the infinite game of the previous section (and therefore to finish the proof of our main result, Theorem~\ref{th:main}).
Let us make first several simple remarks.

First, note that if \textbf{M} has a winning strategy for some $N$, she has also a winning strategy for all larger $N$ (just ignore the vertices that have height greater than $N$). So the words ``there exists some $N$'' can be replaced by ``for every sufficiently large $N$''.

Second, one can scale the game, bounding the total \textbf{M}-weights by some quota $M$ (instead of $1$) and letting $a(\Lambda)$ be some $A$ (also instead of $1$). Then, if \textbf{M} was able to achieve the sum $k$ in the original game, she can use essentially the same strategy in the new game to achieve $kM/A$. For that she should imagine that the actual moves of \textbf{A} are divided by $A$, and multiply by $M$ the moves recommended by the strategy.

Since $k$ in Lemma~\ref{lem:finite} is arbitrary, \textbf{M} can achieve arbitrary large sum even if her weights are limited by arbitrary small constant $\mu>0$ (known in advance); the size $N$ of the tree then depends both on the sum we want to achieve, and on the allowed quota $\mu$. This simple remark allows \textbf{M} to run in parallel several strategies on some subtrees, allocating quotas $\mu_1,\mu_2,\mu_3,\ldots$ to them, where $\sum \mu_i \le 1$ is some converging series, e.g., $\mu_i=2^{-i}$. These strategies achieve sum $1$ in each subtree. It is indeed possible: the flow generated by the adversary can be considered separately on each subtree: if the total flow starting from the root is at most $1$, the flow in every vertex, including the root of a subtree, is also at most $1$. (Note the using $a(\Lambda)<1$ in the root instead of $1$ makes the task of adversary harder, so \textbf{M} can win in every subtree.) These subtrees are chosen as shown on Fig.~\ref{subtrees}. 

Knowing $\mu_1$, we choose the height of the first subtree; knowing the number of leaves in the first subtree and the corresponding $\mu_i$, we choose the appropriate heights for the second layer subtrees (one can choose the same height for all of them to make the picture nicer); then, knowing the number of leaves in all of then, we look at the corresponding $\mu_i$ and select the height for the third layer, etc. The games are played (and won) independently in each subtree. In each subtree there is a path with $\sum m(x)/a(x)\ge 1$, and we can combine these paths into an infinite path starting from the root. 

\section{How to win the finite game}\label{sec:finite-win}

In this section we provide the proof of Lemma~\ref{lem:finite}, therefore finishing the proof of our main result, Theorem~\ref{th:main}. As we have seen, the winning strategy for Mathematician should rely on the on-line nature of the game: if \textbf{M} makes only one move and then stops, Adversary could win by splitting the flow proportional to the weights of the subtrees (see the proof of Proposition~\ref{prop:offline}).

\begin{figure}[h]
\begin{center}
   \includegraphics{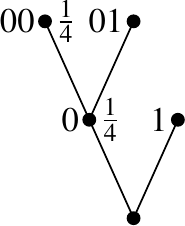}
\end{center}
\caption{First move of \textbf{M} and the reaction of \textbf{A}.}\label{game5}
\end{figure}
To construct the winning strategy for \textbf{M} in the finite game, first let us start with a toy example and show how she can make the sum $\sum_{x\sqsubset w} m(x)/a(x)$ greater than $1$. For this, tree of height $2$ is enough (in fact only some part of it is needed). 

\bigskip

\textbf{M} starts by putting weights $\frac{1}{4}$ to vertices $0$ and $00$ (Figure~\ref{game5}). Then \textbf{A} has to decide how much flow she wants to send to $0$ and $00$. There are several possibilities:

\begin{itemize}
\item The flow to $0$ is small: $a(0)<\frac{1}{2}$. In this case $a(00)$ is obviously also less than $\frac{1}{2}$, so $$\frac{m(\Lambda)}{a(\Lambda)}+\frac{m(0)}{a(0)}+\frac{m(00)}{a(00)}>0+\frac{1}{2}+\frac{1}{2}>1,$$
and this move does not create a winning position for \textbf{A}.

\begin{figure}[h]
\begin{center}
\includegraphics{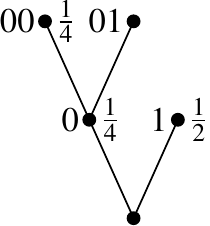}
\end{center}
\caption{A winning move of \textbf{M} in the second case.}\label{game6}
\end{figure}

\item The flow to $0$ is big: $a(0)>\frac{1}{2}$. In this case \textbf{A} may get a winning position (for now).  However, \textbf{M} still can win. Indeed, $a(1)\leq 1-a(0)$ is less than $\frac{1}{2}$ and remains less than $1/2$ forever. Then \textbf{M} puts weight $\frac{1}{2}$ to vertex $1$ (Figure~\ref{game6}), making the sum there greater than $1$, and \textbf{A} cannot do anything.

\begin{figure}[h]
\begin{center}
\includegraphics{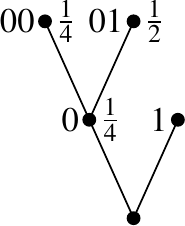}
\end{center}
\caption{A winning move of \textbf{M} in the third case.}\label{game7}
\end{figure}

\item The intermediate case: $a(0)=\frac{1}{2}$. In this case $a(00)$ should be also $\frac{1}{2}$, otherwise the sum in $00$ will still exceed $1$ and \textbf{A} does not get a winning position. But if $a(00)=a(0)=1/2$,  \textbf{M} can put weight $\frac{1}{2}$ to vertex $01$, and \textbf{A} cannot send more than $1/2$ to $01$ (since $1/2$ is already directed to $00$). Then, $$\frac{m(\Lambda)}{a(\Lambda)}+\frac{m(0)}{a(0)}+\frac{m(01)}{a(01)}\geq 0+\frac{1}{4}+\frac{1/2}{1/2}=\frac{5}{4}>1.$$
\end{itemize}

More careful analysis shows that using this idea \textbf{M} can get a winning strategy for $k=17/16$. But we need an arbitrary large $k$ anyway, so we do not go into details, and provide another construction.

The winning strategy for arbitrary $k$ will be recursive: we assume that \textbf{M} has winning strategy for some $k$ and then use this strategy to construct \textbf{M}'s winning strategy for some $k'=k+\varepsilon$, where $\varepsilon>0$. The increase $\varepsilon$ depends on $k$ and is rather small, but has a lower bound $f(k)$ which is a positive continuous function of $k$.

Iterating this construction, we get $k_i$-winning strategies where $k_1=1$ (for $k=1$ the winning strategy is trivial) and
$$ k_{i+1}\ge k_i + f(k_i).$$
We see now that $k_i\to\infty$ and $i\to\infty$; indeed, if $k_i\to K$ for some finite $K$, then $k_{i+1}\ge k_i+f(k_i)\to K+f(K)$, a contradiction.

\clearpage
\begin{wrapfigure}[18]{r}{0cm}
\vspace*{-5mm}
   \hspace*{1cm}\includegraphics{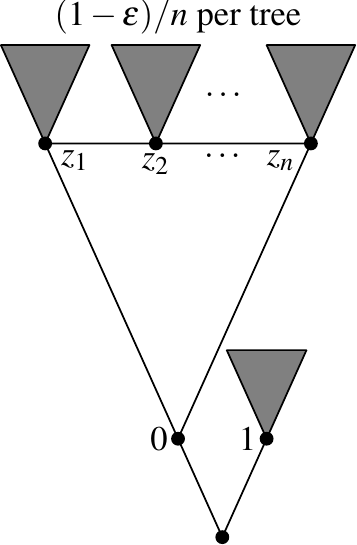}\hspace*{1cm}
\caption{The tree for the inductive $(k+\varepsilon)$-strategy.}\label{game8}
\end{wrapfigure}

To explain the idea of this construction, let us first comment on the toy example explained above (how to get sum greater than $1$). Making her first move, \textbf{M} keeps some reserve that can be later put into the vertex $1$. This possibility creates a constant threat for \textbf{A} that prevents her from directing too much flow to $0$. The same kind of threat will be used in the final construction; again vertex $1$ will be used as ``threat vertex''. If \textbf{A} directs to much flow to the left (vertex $0$), \textbf{M} sees this and uses all the reserve to win in the right subtree.

However, now the strategy is more complicated. There are two main improvements. First, instead of placing some weight in a vertex as before, \textbf{M} uses scaled $k$-strategy in the subtree rooted at that vertex, so the weight is used more efficiently (with factor $k$). This is done both in the left subtree and in the threat vertex $1$. (The subtrees where $k$-strategy is played, are shown in grey in Figure~\ref{game8}.) Second, in the left subtree (of sufficient height) \textbf{M} uses sequentially $n$ vertices $z_1,\ldots,z_n$ (and $n$ corresponding subtrees) for large enough $n$. (We will discuss later how $n$ is chosen.) 

Let us describe the $(k+\varepsilon)$-strategy in more details.

First of all, \textbf{M} puts weight $\varepsilon$ into vertex $0$ (after that \textbf{M} will never add weight there, so vertex $0$ always has weight~$\varepsilon$).\footnote{The weight of vertex $0$ in the strategy is equal to the desired increase in $k$; there are no deep reasons for this choice, but it simplifies the computations.} 

After that \textbf{M} still has weight $1-\varepsilon$ available. It is divided into $n$ equal parts, $(1-\varepsilon)/n$ each. These parts are used sequentially in subtrees with roots $z_1,\ldots,z_n$ (Figure~\ref{game8}). In these subtrees \textbf{M} uses scaled $k$-strategy; the coefficient is $(1-\varepsilon)/n$. In this way \textbf{M} forces \textbf{A} to direct a lot of flow to these $n$ subtrees or lose the $k$-game in one of this subtrees (and therefore lose $(k+\varepsilon)$-game in the entire tree, if the parameters are chosen correctly).

The threat vertex $1$ is used as follows: if at some point (after $i$ games for some $i$) the flow directed by \textbf{A} to $0$ is too large, \textbf{M} changes her strategy and use all the remaining weight, which is $(1-\varepsilon)(1-\frac{i}{n})$, for $k$-strategy in the $1$-subtree (and wins, if the parameters are chosen correctly).

Now we have to quantify the words ``a lot of flow'' and ``too large flow'' by choosing some thresholds. Assume that after $i$ games \textbf{A} directed some weight $d$ to $0$. Then she can use only $1-d$ for the game in the threat vertex. Using $k$-strategy with reserve $(1-\varepsilon)(1-\frac{i}{n})$, \textbf{M} can achieve sum (along some path in the right subtree)
   $$ \frac{k(1-\varepsilon)(1-\frac{i}{n})}{1-d}, $$
so the threshold $d_i$ is obtained from the equation
   $$ \frac{k(1-\varepsilon)(1-\frac{i}{n})}{1-d_i}=k+\varepsilon.\eqno(*)$$
If the flow to the left vertex $0$ is at least $d_i$, \textbf{M} stops playing games in the left subtree and wins the entire game by switching to $k$-strategy in the right subtree and using all remaining weight there.

What happens if \textbf{A} does not exceed the thresholds $d_i$? Then the vertex $0$ adds $\varepsilon/d_i$ to the sum in the $i$-th game, and to win the entire game \textbf{M} needs to get the sum $(k+\varepsilon)-\varepsilon/d_i$ in the $i$-th game. This can be achieved using (scaled) $k$-strategy with weight $(1-\varepsilon)/n$ unless \textbf{A} directs $a_i$ to $z_i$-subtree, where $a_i$ is determined by the equation
$$
k+\varepsilon-\frac{\varepsilon}{d_i}=k\frac{(1-\varepsilon)/n}{a_i}
\eqno(**)
$$
We need to prove, therefore, that for some $\varepsilon$ (depending on $k$) and for large enough $n$ the values $a_i$  determined by $(**)$, where $d_i$ is determined by $(*)$, satify the inequality  
$$ \sum_{i=1}^{n} a_i>1. $$
Then \textbf{A} is unable to direct $a_i$ in $z_i$-subtree for all $i$ and loses the game.

This sum can be rewritten as follows:
$$
\sum_{i=1}^n a_i=
\frac{1}{n}\sum_{i=1}^n\frac{kd_i(1-\varepsilon)}{d_i(k+\varepsilon)-\varepsilon} .
$$
Note that $d_i$ depends only on $k$, $\varepsilon$ and $u=\frac{i}{n}$, so this sum is the Riemann sum for the integral 
$$
\int_0^1 k(1-\varepsilon)\frac{d(u)}{(k+\varepsilon)d(u)-\varepsilon}du
$$ 
where 
$$
d(u)=1-\frac{k}{k+\varepsilon}(1-\varepsilon)(1-u).
$$ 
Note that we integrate a rational function of the form $(Au+B)/(Cu+D)$, so it is not a problem, and we get 
\begin{multline*}
\int_0^1 \frac{(1-\varepsilon)(ku- ku\varepsilon+k\varepsilon+\varepsilon)}{(k+\varepsilon)(u+\varepsilon-u\varepsilon)}du=\\=\frac{k(u+\varepsilon-\varepsilon u)+\varepsilon\log(u(\varepsilon-1)-\varepsilon)}{k+\varepsilon}\biggr|_0^1=\\=\frac{k(1-\varepsilon)+\varepsilon\cdot\log(1/\varepsilon)}{k+\varepsilon}.
\end{multline*} 
Note that for $\varepsilon\to 0$ this expression can be rewritten as 
$$
 1+\varepsilon\cdot(\log(1/\varepsilon)-k-1)/k+O(\varepsilon^2),
$$ 
so for sufficiently small $\varepsilon>0$ this integral will be greater than $1+\varepsilon$, and we can choose $n$ large enough to make the Riemann sum greater than $1$. It is easy to get a positive lower bound for $\varepsilon$ (depending on $k$) and find the corresponding $n$ effectively.

This finishes the proof of Lemma~\ref{lem:finite} and therefore the proof of our main result, Theorem~\ref{th:main}.

\textbf{Remark}. Let us repeat the crucial point of this argument: during the initial phase of the strategy, when \textbf{M}'s reserve is large, \textbf{A} cannot direct a lot of flow into $0$, so the weight $\varepsilon$ placed into this vertex is taken with a large coefficient. Without this threat \textbf{A} could place all the flow in the left subtree, and then the weight $\varepsilon$ would not help: on the contrary, the same weight could have been used $k$ times more efficiently in the subtrees, and we get no increase in $k$.

\section{Improvement: how to find a c. e. set with infinite sum}\label{sec:enumerable}

The construction can be adjusted to guarantee some additional properties of the sequence $\omega$ with $\sum_{x\sqsubset\omega}\m(x)/\A(x)=\infty$.

\begin{theorem}\label{th:enumerable}
There exists a computably enumerable set $X$ such that for its characterstic sequence $\omega_X$ (where $\omega_i=1$ for $i\in X$ and $\omega_i=0$ for $i\notin X$) the sum $\sum_{x\sqsubset\omega_X}\m(x)/\A(x)$ is infinite.
\end{theorem}

\begin{proof}
We start by modification of the finite game of Lemma~\ref{lem:finite}. Let us agree that \textbf{M} should (for each her move) not only achieve a winning position, but also explicitly mark one of the nodes of the tree where the sum is at least $1$ (according to the definition of the winning position). If there are several nodes where the  (current) sum reaches $1$, \textbf{M} can choose any of them. During the game, \textbf{M} can change the marked node, but monotonicity is required: the marked nodes should form an increasing sequence in a coordinate-wise ordering (for each node of a binary tree we consider a sequence of zeros and ones that leads to this node, and add infinitely many trailing zeros; in this way we get an infinite sequence, and when the node changes, the new sequence should be obtained from the previous one by some $0\to 1$ replacements). 

This requirement could be satisfied by minor changes in construction of winning strategy for $(k+\varepsilon)$ game. Note that the winning strategy for $(k+\varepsilon)$ game calls  the winning $k$-strategy for vertices $z_1,\ldots,z_n$, and maybe for the threat vertex. Using induction, we may assume that $k$-strategy satisfies the monotonicity requirement. This is not enough: \textbf{M} needs also to guarantee monotonicity while switching from the $k$-strategy in $z_i$-subtree to the $k$-strategy in $z_{i+1}$-subtree (or in the threat vertex). To achieve this, some precautions are needed. First of all, we choose $z_1,\ldots,z_n$ in such a way that $z_1\le z_2\le\ldots z_n$ coordinate-wise. Moreover, while playing the game above $z_1$, \textbf{M} makes some bits equal to $1$ (according to the winning $k$-strategy in the subtree). These bits cannot be reversed back, but this is not a problem: for example, one can add several $1$s at the end of $z_2$ (to cover all the bits changed while playing above $z_1$), and use a subtree rooted there, then do the similar trick for $z_3$, etc. (see Figure~\ref{monotone2}). Finally, the same can be done for the threat vertex.

\begin{figure}[h]
\begin{center}
   \includegraphics{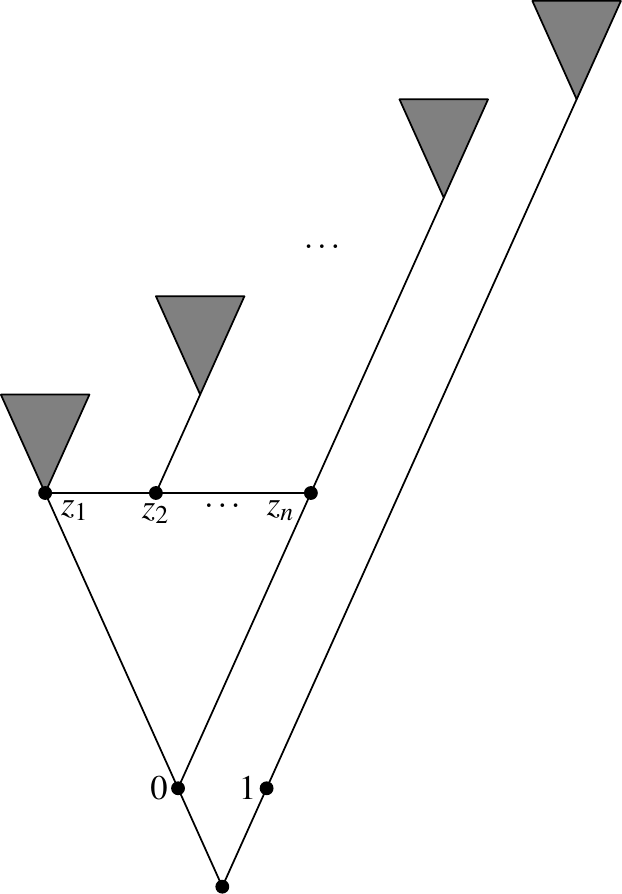}
 \end{center}
\caption{Special precautions needed to preserve the monotonicity during the induction step.}\label{monotone2}
\end{figure}

Infinite game of Lemma~\ref{lem:infinite} can also be adjusted. Here \textbf{M} should after each move maintain a \emph{current branch}, an infinite path in the binary tree that contains only finitely many ones (so it is essentially a finite object and can be specified by \textbf{M} explicitly).  The current branch may change during the game but in a monotone way: it should increase \emph{coordinate-wise}. In other words, if the previous branch went right at some level, the next one should do it too (at the same level). This monotonicity requirement guarantees that there exists a limit branch, and \textbf{M} wins the (infinite) game if the sum is infinite along this branch.

\begin{figure}[h]
\begin{center}
   \includegraphics{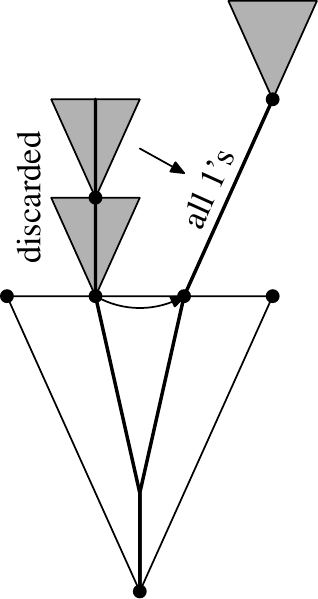}
 \end{center}
\caption{Subtrees where the games are started. When a marked leaf changes, all the subtrees above it are abandoned, and new subtree is chosen with all $1$s inbetween.}\label{mon_subtrees}
\end{figure}

We claim that \textbf{M} has a computable winning strategy in this game. Knowing this, we can easily construct an enumerable set required by Theorem~\ref{th:enumerable}. Again we use the computable winning strategy against a ``blind'' adversary that enumerates from below the values of the continuous a priori semimeasure. Then the behavior of the adversary is computable, the behavior of the computable winning strategy is also computable, and the limit branch will be a characteristic sequence of a (computably) enumerable set.

It remains to explain how one can combine winning strategies for finite games (modified) to get a winning strategy for the infinite game. We cannot run the strategies on subtrees in parallel as we did before, because the candidate branches provided by the strategies at the same level will not be related, and switching from one game to another will violate the ordering condition. Instead, we start first the game in the root subtree. The strategy makes some move, in particular, marks some leaf of this subtree (``current candidate''). Then we start the strategy on a subtree that is above this marked leaf. This strategy marks some leaf in this subtree, and we start a third game above it, etc. (See Figure~\ref{mon_subtrees}.)

At some point one of these strategies may change its marked leaf. Then all the games started above this (now discarded) leaf are useless, and we start a new game above the new marked leaf.  To satisfy the monotonicity condition, we should start the new game high enough and put $1$s in all positions below the starting point of the new game. This will guarantee that all $1$s that were already in the current branch will remain there. (We assume that at every moment only finitely many games are started, and the current branch has only finitely many ones.)

One can see that in the limit we still have a branch with infinite sum. Indeed, in the root game the current marked leaf can only increase in the coordinate-wise ordering, and only finitely many changes of marked leaf are possible. Therefore, some leaf will remain marked forever. The game started above this leaf will never be discarded, but the leaves marked in this game may change (monotonically).  This happens finitely many times, and after that the marked leaf remains the same, the game above it is never discarded, etc.

The monotonicity is guaranteed both for the elements inside the tree where the marked leaf changed (according to the monotonicity for finite games) and for outside elements (since we replace the bits in the discarded parts by $1$s only).
\end{proof}


\end{document}